\documentclass[12pt,journal,onecolumn,draftcls]{IEEEtran}

\ifx\noeditingmarks\undefined
  \newcommand{\pgwrapper}[2]{\textred{#1: #2}}
\else
  \newcommand{\pgwrapper}[2]{}
\fi
\usepackage{soul}
\usepackage[normalem]{ulem}
\usepackage{framed}
\usepackage{bbm}
\usepackage{epsfig}
\usepackage{times}
\usepackage{float}
\usepackage{afterpage}
\usepackage{amsmath}
\usepackage{amstext}
\usepackage{amssymb,bm}
\usepackage{latexsym}
\usepackage{color}
\usepackage{graphicx}
\usepackage{amsmath}
\usepackage{amsthm}
\usepackage[center]{caption}
\usepackage{pstricks}
\usepackage{caption}
\usepackage{subcaption}
\usepackage{booktabs}
\usepackage{multicol}
\usepackage{lipsum}
\usepackage{verbatim}
\usepackage{mdframed}
\usepackage{booktabs,makecell}
\usepackage{hyperref}

\newtheorem{remark}{Remark}

\newtheorem{proposition}{Proposition}
\newtheorem{lemma}{Lemma}

\newtheorem{theorem}{Theorem}

\newcommand{\bfX}{\mathbf{X}}
\newcommand{\bfY}{\mathbf{Y}}
\newcommand{\bfy}{\mathbf{y}}
\newcommand{\bfZ}{\mathbf{Z}}

\newcommand{\bfW}{\mathbf{W}}
\newcommand{\bfw}{\mathbf{w}}
\newcommand{\CAP}{\mathsf{C}}
\newcommand{\snr}{\mathsf{S}}
\newcommand{\inr}{\mathsf{I}}
\newcommand{\Fopt}{F_{\widetilde{X}}^{\text{opt}}}

\newcommand{\xb}{{x+\inr+2\sqrt{x \inr} \cos \theta }}
\newcommand{\dF}{{\ dF_{\widetilde{X}}(x) }}
\newcommand{\dy}{{\ dy }}
\newcommand{\kernel}{{e^{-(y+\xb)}I_0\left( 2\sqrt{y}\sqrt{\xb}\right) }}

\begin{document}

\title{On the Capacity of the AWGN Channel with Additive Radar Interference}
\author{
\IEEEauthorblockN{Sara Shahi, Daniela Tuninetti, and Natasha Devroye\\}
\IEEEauthorblockA{%
ECE, University of Illinois at Chicago, Chicago IL 60607, USA.\\
Email: {\tt sshahi7, danielat, devroye @uic.edu}}%
}

\maketitle

\begin{abstract}
This paper investigates the capacity of a communications channel that, in addition to additive white Gaussian noise, also suffers from interference caused by a co-existing radar transmission. 
The radar interference (of short duty-cycle and of much wider bandwidth than the intended communication signal) is modeled as an additive term whose amplitude is known and constant, but whose phase is independent and identically uniformly distributed at each channel use. The capacity achieving input distribution, under the standard average power constraint, is shown to have independent modulo and phase.
The phase is uniformly distributed in $[0,2\pi]$.
The modulo is discrete with countably infinitly many mass points, but only finitely many in any bounded interval.
From numerical evaluations, a proper-complex Gaussian input is seen to perform quite well for weak radar interference.
%
We also show that for very large radar interference, capacity is equal to $\frac{1}{2}\log\left(1+\snr \right)$  and 
a proper-complex Gaussian input achieves it.
It is concluded that the presence of the radar interference results in a loss of half of the degrees of freedom compared to an AWGN  channel without radar interference. 
\end{abstract}

\section{Introduction}
\label{sec:intro}
Shortage of spectrum resources, coupled with the ever increasing demand for commercial services, necessitates a more sensible bandwidth allocation policy. In 2012, the President's council of Advisors on Science and Technology  published a report 
that recommended the release of portions of government radar bands (e.g., 3550-3700 MHz) to be shared with commercial wireless services. A new Citizens Broadband Radio Service (CBRS) for shared wireless broadband in the  3550-3700 MHz band has also been established in 2015. Since then, several national funding agencies 
have launched 
research programs to encourage research in this area.

To understand how these two very different systems should  best share the spectrum, it is useful to have an idea of the fundamental information theoretic performance limits of channels in which two systems co-exist.
In this work, the capacity of a white additive Gaussian noise communications channel, which in addition to noise, suffers interference from a radar transmission, is studied. This extends the authors' prior work in ~\cite{shahi-allerton}. 
In the channel model considered, the interfering radar transmission is modeled to be additive, but not Gaussian. Rather, it is modeled as a constant (thus known) amplitude signal, but with unknown and uniformly distributed phase at each channel use. A detailed justification of this model can be found in~\cite{Salim_Conf_RadarPMF_2016}.
%
The capacity of an additive Gaussian noise channel under an average power constraint is well known: the optimal input is Gaussian of power equal to the power constraint.
However, since the channel studied here is no longer Gaussian, several questions emerge: 
(i) what is the capacity of this channel and how does it differ from that of a Gaussian noise channel (without the radar interference), and 
(ii) what 
input achieves the capacity. In this paper we aim to address both these questions. 
\subsection{Past Work} 
The capacity of channels with additive noise under various input constraints has been studied. 

In~\cite{ihara} Ihara bounds the capacity of additive, but not necessarily Gaussian, noise channels. Applying Ihara's upper bound
to our channel model yields a bound that grows 
as the radar signal amplitude increases. This bound is not tight because the capacity of our channel is upper bounded by the capacity of the classical power-contained Gaussian noise channel without radar interference.

In~\cite[Theorem~1]{1302319}, it was shown that for any memoryless additive noise channel with a second moment/power constraint on the input, the 
rate achievable by using a white Gaussian input never incurs a loss of more than half a bit per real dimension with respect to the capacity. This implies that one can obtain a capacity upper bound for a complex-valued additive noise channel by adding 1~bit to the rate attained with a proper-complex Gaussian input for the same channel. In our channel model however, we show that in the large radar interference regime, the capacity is indeed achieved by a proper-complex Gaussian input and hence the bound given in~\cite[Theorem~1]{1302319} is loose for large radar interference. 

In the seminal work by Smith~\cite{smith}, 
it was shown that the capacity of a real-valued white Gaussian noise channel with peak~amplitude and average~power constraints is achieved by a discrete input with a finite number of mass points. This is in sharp contrast to the Gaussian input that achieves the capacity when the amplitude constraint is dropped. 
In~\cite{bar-david}, the authors also considered complex-valued Gaussian noise with average~power and peak~amplitude constraints and derived the  optimal input distribution characteristics. In particular, they showed that the capacity of the complex-valued Gaussian noise channel under average~power and peak~amplitude constraints is achieved by a complex-valued input with independent amplitude and phase; the optimal phase is uniformly distributed in the interval $[0,2\pi]$, and the optimal amplitude is discrete with finitely many mass points. 

Later, 
the optimality of a discrete input under peak~amplitude constraint was shown to hold for a wide class of real-valued additive noise channels~\cite{aslan}. As for the complex-valued additive channels, \cite{meyn}  showed that for  certain additive complex-valued channels with average~power and peak~amplitude constraints, the optimal input has discrete modulo.
%
Moreover, recently it was shown in~~\cite{faycal}, that,
under an average~power constraint and certain `smoothness' conditions on the noise distribution, 
the only real-valued additive noise channel whose capacity achieving input is continuous is the Gaussian noise channel.

The model considered in this paper is a complex-valued additive noise channel with an average~power constraint. 
When we transform the mutual information optimization problem over a bivariate (modulo and phase) input distribution into one over a univariate (modulo only) input distribution, the equivalent channel (i.e., the channel  kernel $K(x,y)$, which is formally defined in~\eqref{eq:kernel def}) is no longer additive. For this equivalent non-additive channel, we can \textit{not} proceed as per the steps preceding~\cite[eq.(4)]{faycal}. This is so, because~\cite[eq.(4)]{faycal}, heavily relies on certain integrals being convolution integrals and thus passing to the Fourier domain to study/infer certain properties of the optimal input distribution. In non-additive channels this is not possible.

In this respect, 
our approach is similar to that of~\cite{bar-david} and we closely follow the steps within that. In particular, the trick used in~\cite{bar-david} to reduce the two dimensional optimization problem into a one dimensional optimization problem helps us to avoid the use of an identity theorem for multiple dimensions. In fact, it was shown in~\cite{boche} that the application of an identity theorem in several variables is not just a straightforward generalization of one variable. 
%
%
%
%


Extensions of Smith's work~\cite{smith} to Gaussian channels with various fading models, possibly MIMO, are known in the literature but are not reported here because they are not directly relevant.

In~\cite{NN-RadarConf2016,NN-Globecom2016} a subset of the authors studied the uncoded average symbol error rate performance of the same 
channel model considered here. Two regimes of operation emerged.
In the low Interference to Noise Ratio (INR) regime, it was shown that the optimal decoder is 
a minimum Euclidean distance decoder, as for Gaussian noise only;
while in the high INR regime, radar interference estimation and cancellation is optimal. Interestingly, 
in the process of canceling the radar interference at high INR, part of the useful signal is also removed, and after cancellation the equivalent output is real-valued (one of the two real-valued dimensions of the original complex-valued output is lost). We shall observe the similar `half degrees of freedom' loss for the capacity of this channel. 
%
%
%
\subsection{
Contributions} 
The capacity of the  channel model proposed here has not, to the best of our knowledge, been studied before and provides a new model for bounding the fundamental limits of a communication system in the presence of radar interference. 
Likewise, in the literature on the co-existence of radar and communications channels, we are not aware of any 
capacity results. 
Our contributions thus lie 
in the study of the capacity of this novel channel model, in which we show that the optimal input distribution has independent modulo and phase.
The phase is uniformly distributed in $[0,2\pi]$.
The modulo is discrete with countably infinite many mass points, but only finitely many in any bounded interval.

By upper bounding the output entropy by the cross entropy of the output distribution and an output distribution induced by Gaussian input, we show that very high radar interference results in a loss of half the degrees of freedom compared to an interference-free channel and that a Gaussian input is optimal in the high interference regime.

We also show achievable rates.
The Gaussian input is seen to perform very well for weak radar interference, where it closely follows the upper bound in~\cite{ihara}. 
%
%
%
  We numerically find some suboptimal  inputs for the regime $0<\alpha:=\frac{\inr_{(\text{dB})}}{\snr_{(\text{dB})}}<2$, where $\inr$ and $\snr$ are the interference and signal to noise ratios respectively,  which  perform better than the Gaussian input.

\subsection{
Paper organization} 
The paper is organized as follows.
Section~\ref{sec:model} introduces the channel model. 
Section~\ref{subsec:main} derives our main result.
Section~\ref{sec: high inr} finds the capacity for large INR regime.
Section~\ref{sec:num} provides numerical results.
Section~\ref{sec:concl} concludes the paper.
Proofs can be found in the Appendix.


\section{Problem formulation and preliminary results}
\label{sec:model}

Next, boldface letters indicate complex-valued random variables, while
lightface letters real-valued ones. Capital letters represent the random variables and the lower case letters represent their realization. In addition, $\widetilde{X}$ and $\angle \bfX$ are respectively used to indicate the modulus squared, $|\bfX|^2$, and phase of the complex-valued random variable $\bfX$.  We also define the following notation 
\begin{align*}
\mathbb{R}_+&:=\{x: x\geq 0\},\\
\mathbb{C}_+&:=\{z:z\in \mathbb{C}, \ \Re(z)>0\},
\end{align*}
for the set of non-negative real line and right half plane in the complex domain, respectively.
\subsection{System model}
We model the effect of a high power, short duty cycle radar pulse at the receiver of a narrowband data communication system as 
\begin{align}
 \bfY &=\bfX+\bfW,
\label{eq:additive ch output}
 \\
 \bfW &=\sqrt{\inr}e^{j\Theta_{\inr}}+\bfZ,
\label{eq:additive ch noise}
\end{align}
where 
$\bfY$ is the channel output, 
$\bfX$ is the input signal subject to the average power constraint $\mathbb{E}[|\bfX|^2]\leq \snr$, 
$\Theta_{\inr}$ is the random phase of the radar interference uniformly distributed in $[0,2\pi]$, and 
$\bfZ$ is a zero-mean proper-complex unit-variance Gaussian noise. 
The random variables $(\bfX,\Theta_{\inr},\bfZ)$ are mutually independent. 
$\Theta_{\inr}$ and $\bfZ$ are independent and identically distributed over the channel uses, that is, the channel is memoryless. 
Our normalizations imply that $\snr$ is the average Signal to Noise Ratio (SNR) 
while $\inr$ is the average Interference to Noise Ratio (INR). 
We assume $\inr$ to be fixed 
and thus known. 
For later use, the distribution of the additive noise in~\eqref{eq:additive ch noise} is given by
\begin{align}
f_\bfW(\bfw) 
  &= \mathbb{E}_{\Theta_{\inr}}
  \left[ \frac{e^{-|\bfw-\sqrt{\inr} e^{j\Theta_{\inr}}|^2}}{\pi} \right]
= \frac{e^{-|\bfw|^2-\inr}}{\pi} I_0\left(2\sqrt{\inr|\bfw|^2}\right),
\label{eq:noise pdf in w}
\end{align}
where $I_0(\bfw)=\mathbb{E}[e^{\bfw\cos(\Theta_{\inr})}]\in[1,e^{|\bfw|}]$ for
$\bfw\in\mathbb{C}$ is the zero-order modified Bessel function of the first kind.
The channel transition probability is thus 
\begin{align}
f_{\bfY|\bfX}(\bfy|\mathbf{x})
=f_{\bfW}(\bfy-\mathbf{x}), \quad (\mathbf{x},\bfy)\in\mathbb{C}^2.
\label{eq:additive ch noise pdf}
\end{align}

\subsection{Channel Capacity} \label{sec:chcap}
Our goal is to characterize 
the capacity of the memoryless channel in~\eqref{eq:additive ch output}-\eqref{eq:additive ch noise} given by
\begin{align}
\CAP(\snr, \inr) = \sup_{F_\bfX: \mathbb{E}[|\bfX|^2]\leq \snr}I(\bfX;\bfY),
\label{eq:cap def main}
\end{align}
where $F_\bfX$ is the cumulative distribution function of  $\bfX$, and $I(\bfX;\bfY)$ denotes the mutual information between random variables $\bfX$ and $\bfY$ in \eqref{eq:additive ch output}.

We aim to show that the supremum in~\eqref{eq:cap def main} is actually attained by a {\it unique} input distribution, for which we want to derive its structural properties. Before we continue however, we rewrite the optimization for the original channel~\eqref{eq:additive ch output} (involving the real and the imaginary part of the input) in a way that allows optimization with respect to a univariate distribution only.

By following steps similar to those in~\cite[Section II.B]{bar-david}, we can show that an optimal input distribution induces $\widetilde{Y}$ 
 and $\angle{\bfY}$ independent given $\bfX$, with $\angle{\bfY}$ uniformly distributed over the interval $[0,2\pi]$; such an output distribution can be attained by the uniform distribution on $\angle{\bfX}$ and by $\angle{\bfX}$ independent of $\widetilde{X}$
 ; therefore, it is convenient for later use to denote the channel transition probability $f_{\widetilde{Y}|\widetilde{X}}(y|x)$ as the kernel $K(x,y)$ given by (see Appendix \ref{app:kernel derivation})
\begin{subequations}
\begin{align}
K(x,y)
 & :=f_{\widetilde{Y} \mid \widetilde{X}}(y|x)
\\&= \int_{|\theta|\leq\pi} \frac{e^{-\inr-\xi(\theta;x,y)}}{2\pi} 
     I_0\left(2\sqrt{\inr\ \xi(\theta;x,y)}\right) \ d\theta,    
\\
\xi(\theta;x,y)&:=y+x-2\sqrt{yx}\cos(\theta) \geq 0,  \ (y,x)\in\mathbb{R}_{+}^2.
\end{align}
\label{eq:kernel def}
\end{subequations}
Since the random variables $\widetilde{X}$ 
  and $\widetilde{Y}$
   are connected through a channel with kernel $K(x,y)$, an input distributed as $F_{\widetilde{X}}$ results in an output with probability distribution function (pdf)%
\footnote{
The pdf $f_{\widetilde{Y}}(y;F_{\widetilde{X}})$ in~\eqref{eq:output pdf def} exists since the kernel in~\eqref{eq:kernel def} is a continuous and bounded (see~\eqref{eq:kernel bounds}) function and thus integrable.
}
\begin{align}
f_{\widetilde{Y}}(y;F_{\widetilde{X}}):=\int_{x\geq 0}K(x,y) dF_{\widetilde{X}}(x),
\quad y\in\mathbb{R}_{+}.
\label{eq:output pdf def}
\end{align}
We stress the dependence of the output pdf on the input distribution $F_{\widetilde{X}}$ by adding it as an `argument' in $f_{\widetilde{Y}}(y;F_{\widetilde{X}})$. 

Finding the channel capacity in~\eqref{eq:cap def main} can thus be equivalently expressed as the following optimization over the distribution of a non-negative random variable $\widetilde{X}$
\begin{align}
\CAP(\snr, \inr) + h(|\bfW|^2)
&= \sup_{F_{\widetilde{X}}: \mathbb{E}[\widetilde{X}]\leq \snr} h(\widetilde{Y};F_{\widetilde{X}}),
\label{eq:cap def}
\end{align}
where $h(\widetilde{Y};F_{\widetilde{X}})$ is the output differential entropy given by
\footnote{
The entropy $h(\widetilde{Y};F_{\widetilde{X}})$ in~\eqref{eq:output entropy def0} exists since the output pdf in~\eqref{eq:output pdf def} is a continuous and bounded (see~\eqref{eq:output pdf bounds}) function and thus integrable.
}
\begin{align}
&h(\widetilde{Y};F_{\widetilde{X}})
   = \int_{y\geq 0} f_{\widetilde{Y}}(y;F_{\widetilde{X}}) \log\frac{1}{f_{\widetilde{Y}}(y;F_{\widetilde{X}})} \ dy.
\label{eq:output entropy def0}
\end{align}
We express $h(\widetilde{Y};F_{\widetilde{X}})$ in~\eqref{eq:output entropy def0} as
\begin{align}
h(\widetilde{Y};F_{\widetilde{X}})
&= \int_{y\geq 0} \int_{x\geq 0} 
     K(x,y)  \log\frac{1}{f_{\widetilde{Y}}(y;F_{\widetilde{X}})} \ d F_{\widetilde{X}}(x) \ dy
\notag
\\&= \int_{x\geq 0} 
     h(x; F_{\widetilde{X}})\ d F_{\widetilde{X}}(x),
\label{eq:output entropy def}
\end{align}
where we defined the {\it marginal entropy} $h(x;F_{\widetilde{X}})$ as
\footnote{
The marginal entropy $h(x; F_{\widetilde{X}})$ in~\eqref{eq:marginal output entropy def} exists since the involved functions are integrable by~\eqref{eq:kernel bounds} and~\eqref{eq:output pdf bounds}.
}
\begin{align}
h(x; F_{\widetilde{X}})
  &:= \int_{y\geq 0} \!\!\!
      K(x,y)\log\frac{1}{f_{\widetilde{Y}}(y;F_{\widetilde{X}})} \ dy, \ x\in\mathbb{R}_{+},
\label{eq:marginal output entropy def}
\end{align}
and where the order of integration in the line above~\eqref{eq:output entropy def} can be swapped by Fubini's theorem.

For later use, we note that the introduced functions can be bounded as follows:
for the kernel in~\eqref{eq:kernel def}
\begin{align}
e^{-(y+x+\inr)} \leq K(x,y) \leq 1, \quad (x, y)\in \mathbb{R}_{+}^2;
\label{eq:kernel bounds}
\end{align}
for the output pdf in~\eqref{eq:output pdf def}
\begin{align}
e^{-(y+\inr+\beta_{F_{\widetilde{X}}})}\leq f_{\widetilde{Y}}(y;F_{\widetilde{X}}) \leq 1, \quad y\in \mathbb{R}_{+},
\label{eq:output pdf bounds}
\end{align}
where $\beta_{F_{\widetilde{X}}}$ is defined and bounded (by using Jensen's inequality together with the power constraint) as
\begin{align}
0\leq \beta_{F_{\widetilde{X}}}:=-\ln\left(\int_{x\geq 0} e^{-x}dF_{\widetilde{X}}(x)\right) \leq \snr; 
\label{eq:betaF def}
\end{align}
for the marginal entropy in~\eqref{eq:marginal output entropy def}
\begin{align}
0\leq 
h(x; F_{\widetilde{X}}) 
\leq \mathbb{E}[\widetilde{Y}|\widetilde{X}=x]+\inr+\beta_{F_{\widetilde{X}}},
\quad x\in \mathbb{R}_{+},
\label{eq:marginal output entropy bounds}
\end{align}
where the conditional mean of $\widetilde{Y}$ 
given $\widetilde{X}$ 
 is 
\begin{align}
\mathbb{E}[\widetilde{Y}|\widetilde{X}=x] 
&= 
x+\inr+1, \quad x\in \mathbb{R}_{+}.
\label{eq:conditional output power}
\end{align}

\subsection{Trivial Bounds}
\label{subsec:trivial bounds}
Trivially, one can lower bound the capacity in~\eqref{eq:cap def main}
by treating the radar interference as a Gaussian noise and obtain
\begin{align}
\log\left(1+\frac{\snr}{1+\inr}\right) \leq \CAP(\snr,\inr),
\label{eq:cap trivial lower}
\end{align}
and upper bound it as
\begin{align}
\CAP(\snr,\inr) 
\leq \max_{F_\bfX: \mathbb{E}[|\bfX|^2]\leq \snr}I(\bfX;\bfY,\Theta_{\inr})
= \log\left(1+\snr\right),
\label{eq:cap trivial upper 1}
\end{align}
or 
from Ihara's work~\cite{ihara} as
\begin{align}
\CAP(\snr,\inr) \leq \log\left(\pi e(1+\snr+\inr) \right)-h(\bfW),
\label{eq:cap trivial upper 2}
\end{align}
or from Zamir and Erez's work~\cite[Theorem~1]{1302319}, as
\begin{align}
\CAP(\snr,\inr) \leq I(\bfX_G;\bfY)+\log(2),
\label{eq:cap trivial upper 3}
\end{align}
where $I(\bfX_G;\bfY)$ is the achievable rate with a proper-complex Gaussian input that meets the power constraint with equality.

We shall use these bounds later to benchmark the achievable performance in Section~\ref{sec:num}.

\section{Main Result}
\label{subsec:main}

We are now ready to state our main result: a characterization of the structural properties of the optimal input distribution in~\eqref{eq:cap def main}, in relation
to the problem in~\eqref{eq:cap def}.

\begin{theorem}
The optimal input distribution in~\eqref{eq:cap def main} is unique and has independent modulo and phase.
The phase is uniformly distributed in $[0,2\pi]$.
The modulo is discrete with countably infinite many mass points, but only finitely many in any bounded interval. 
\end{theorem}
\begin{IEEEproof}
As argued in Section~\ref{sec:chcap}, an optimal input distribution in~\eqref{eq:cap def main} has $\angle \bfX$ uniformly distributed in $[0,2\pi]$ and independent of $\widetilde{X}$. The modulo squared $\widetilde{X}$, solves the problem in~\eqref{eq:cap def}, whose supremum is attained by the {\it unique} input distribution $\Fopt$, 
because (see~\cite[Theorem 1]{faycal-shamai}):
\begin{enumerate}

\item
the space of input distributions $\mathcal{F}$ is compact and convex (see~\cite[Theorem 1]{faycal-shamai}); $\mathcal{F}$ is given by
\begin{subequations}
\begin{align}
\mathcal{F}:=
\Big\{
  & F_{\widetilde{X}}: F_{\widetilde{X}}(x)=0, \ \forall x<0, 
\label{input dist constraint:non-negativity}
\\& d F_{\widetilde{X}}(x)\geq 0, \ \forall x\geq 0, 
\label{input dist constraint:a valid distribution 1}
\\& \int_{x\geq 0} 1  \cdot d F_{\widetilde{X}}(x) =1,
\label{input dist constraint:a valid distribution 2}
\\& L(F_{\widetilde{X}}):=\int_{x\geq 0} x  \cdot d F_{\widetilde{X}}(x) - \snr\leq 0 \label{input dist constraint:power constraint}
\Big\},
\end{align}
\label{input dist constraint}
where the various constraints are:
\eqref{input dist constraint:non-negativity} for non-negativity,
\eqref{input dist constraint:a valid distribution 1} and~\eqref{input dist constraint:a valid distribution 2} for a valid input distribution, and
\eqref{input dist constraint:power constraint} for the average power constraint; and
\end{subequations}
\item
The differential entropy $h(\widetilde{Y};F_{\widetilde{X}})$ in~\eqref{eq:output entropy def} is a weak$^{\star}$ continuous (see Appendix~\ref{appendix:continuity proof}) and strictly concave (see Appendix~\ref{appendix:concavity proof}) functional of the input distribution $F_{\widetilde{X}}$.
\end{enumerate}
From this and by Smith's approach~\cite{smith}, the solution of the optimization problem in~\eqref{eq:cap def} 
is equivalent to the solution of 
\begin{subequations}
\begin{align}
h^{\prime}_{\Fopt}(\widetilde{Y};F_{\widetilde{X}})
-\lambda 
L^{\prime}_{\Fopt}(F_{\widetilde{X}})\leq 0, \text{ for all } F_{\widetilde{X}}\in \mathcal{F},
\\
\lambda\geq 0 : L(\Fopt)=0,
\end{align}
\label{eq:equivalent optimization}
\end{subequations}
where the functional $L(.)$ was defined in~\eqref{input dist constraint:power constraint},
and where the prime sign along with the subscript $\Fopt$
denotes the weak$^{\star}$ derivative of the function $h(\widetilde{Y};F_{\widetilde{X}})$ at $\Fopt$~\cite{smith}
(see Appendix~\ref{app:weak differentiability}).

The conditions in~\eqref{eq:equivalent optimization} can be equivalently expressed as the necessary and sufficient Karush-Kuhn-Tucker (KKT) condition: for some $\lambda\geq 0$
\begin{align}
h(x; \Fopt) \leq  h(\widetilde{Y}; \Fopt)
 +\lambda(x-\snr), \quad \forall x\in \mathbb{R}_+,
\label{eq:kkt}
\end{align}
where equality in~\eqref{eq:kkt} holds {\it only} at the points of increase of $\Fopt$  (see Appendix~\ref{app:kkt-equivalence}). 

At this point, as it is usual in these types of problems~\cite{smith}, the proof follows by ruling out different types of distributions. 
A distribution takes one of the following forms:
\begin{enumerate}
\item \label{continuous} 
Its support contains an infinite number of mass points in some bounded interval;
\item \label{finite} 
It is discrete with finitely many mass points; or
\item \label{remaining} 
It is discrete with countably infinitely many mass points but only a finite number of them in any bounded interval.
\end{enumerate}
Next, we will rule out cases~\ref{continuous} and~\ref{finite} by contradiction.

{\bf Rule out case~\ref{continuous}} ($\Fopt$ has an infinite number of mass points in some bounded interval). We prove that this case requires the inequality in~\eqref{eq:kkt} to hold with equality for all $x\geq 0$;  we then prove this to be impossible.

We start by stating the following proposition, the proof of which is given in Appendix~\ref{app:proof prop lambda}.
\begin{proposition}\label{prop:lambda}
The optimal Lagrange multiplier $\lambda^\text{opt}(\snr)$, 
which represents the weak$^{\star}$ derivative of the capacity $\CAP(\snr,\inr)$ with respect to $\snr$, must satisfy $0 < \lambda^\text{opt}(\snr) < 1$ for all $\snr>0$.
\end{proposition}
For the feasible range $0<\lambda<1$, we re-write the KKT condition in~\eqref{eq:kkt} by following the recent work~\cite{faycal}.
Given the conditional output power expressed as in~\eqref{eq:conditional output power},
we can write 
\begin{align}
x-\snr 
  &= \int_{y\geq 0} \!\!\!
  \left(y-(1+\inr+\snr)\right)  K(x,y)  \ dy, \ \forall x\in\mathbb{R}_+.
\label{eq:new way to write pc}
\end{align}
With~\eqref{eq:new way to write pc}, the KKT condition in~\eqref{eq:kkt} reads: 
there exists a constant $0< \lambda < 1$ such that 
\begin{align}
g(x,\lambda)  \leq h(\widetilde{Y};\Fopt)=\text{constant for all $x\in \mathbb{R}_+$},
\label{eq:new KKT}
\end{align}
with equality only at the points of increase of $\Fopt$,
and where
\begin{subequations}
\begin{align}
g(x,\lambda)  
  &:= \int_{y\geq 0}  K(x,y)\log\left(\frac{\lambda e^{-\lambda y}}{f_{\widetilde{Y}}(y;\Fopt)} \right) \ dy
\label{eq:def g funt lx}
\\& 
  +\log\frac{1}{\lambda}+\lambda(1+\inr+\snr).
\label{eq:def g funt rx}
\end{align}
\label{eq:def g funt}
\end{subequations}

%
We show next that~\eqref{eq:new KKT} can not be satisfied if $\Fopt$ contains an infinite number of mass points in some bounded interval. This step is accomplished by showing that the function $g(x,\lambda), \ x\in\mathbb{R}_+,$ in~\eqref{eq:def g funt} can be extended to the complex domain and that $g(z,\lambda), \ z\in\mathbb{C}_+,$ is analytic. 
\begin{remark}
In this type of analysis, we only require the analyticity of the function $g(z,\lambda)$ over a region in the complex domain which contains the non-negative real line. Hence, it is sufficient to prove the analyticity of $g(z,\lambda)$ over a strip around the non-negative real line but we prove it over the entire right half plane (see Appendix~\ref{app:analyticity}).
\end{remark}
Since the analytical function $g(z,\lambda)$ is equal to a constant at the set of points of increase of $\Fopt$ and since the set of points of increase of $\Fopt$ has an accumulation point (by the Bolzano Weierstrass Theorem~\cite{Lang-complex}), by the Identity Theorem~\cite{Lang-complex}, we conclude that $g(z,\lambda)=\text{constant}, \ \forall z\in \mathbb{C}_+$. As the result $g(x,\lambda)=\text{constant}, \ \forall x\in \mathbb{R}_+$. 
%
%
%
%
One solution, and the \textit{only solution} due to invertibility of the kernel $K(x,y)$ (see Appendix~\ref{app:invertibility}), for  $g(x,\lambda)$ to be a constant and not to depend on $x$ is that the function that multiplies the kernel in the integral in~\eqref{eq:def g funt lx} is a constant (in which case $\int_{y\geq 0} K(x,y) \ dy=1$ for all $x\in \mathbb{R}_+$). 
For this to happen, we need
\begin{align}
f_{\widetilde{Y}}(y;\Fopt)= \lambda e^{-\lambda y} , \ \forall y\in \mathbb{R}_+,
\end{align}
or in other words, we need  the output $\bfY$ to be  a zero-mean proper-complex Gaussian random variable.
Such an output in additive models is only possible if the noise is Gaussian, which is only possible if $\inr=0$.
Therefore, unless $\inr=0$, is it impossible for $\Fopt$ to have an accumulation point and therefore $\Fopt$ must have finitely many masses in any bounded interval.
Thus, we ruled out case~\ref{continuous}.

{\bf Rule out case~\ref{finite}} ($\Fopt$ has a finite number of points).
We again proceed by contradiction.
We assume that the number of mass points is finite, say given by an integer $M < +\infty$,
with optimal masses located at $0\leq x_1^\star< \ldots< x_M^\star <\infty$ and each occurring with probability $p_1^\star,\ldots,p_M^\star,$ respectively. Note that the superscript $\star$ is used to emphasize the optimality of the parameters. Then the output pdf corresponding to this specific input distribution is
\begin{align}
 f_{\widetilde{Y}}(y;{ \Fopt})
  &=\sum_{i=1}^M p_i^\star K(x_i^\star,y)\notag
\\&=\sum_{i=1}^M p_i^\star \int_{|\theta|\leq \pi}\frac{e^{-(y+x_i^\star+\inr+2\sqrt{x_i^\star\inr}\cos \theta)}}{2\pi}
 \quad \cdot
 I_0\left(2\sqrt{y(x_i^\star+\inr+2\sqrt{x_i^\star\inr}\cos \theta)}\right)d\theta,
 \label{eq:output pdf with equiv.kernel}
\end{align}
where the expression in~\eqref{eq:output pdf with equiv.kernel} is based on an equivalent way to write the kernel in~\eqref{eq:kernel def} (see eq.\eqref{eq:kernel alternative def} in Appendix~\ref{app:kernel derivation}).
With~\eqref{eq:output pdf with equiv.kernel}, 
one can bound the marginal entropy 
in~\eqref{eq:marginal output entropy def} as 
\begin{subequations}
\begin{align}
 -h(x; { \Fopt})&=\int_{y\geq 0} K(x,y) \log f_{\widetilde{Y}}(y;\Fopt) dy\\
&\leq -\left(x\!+\!\inr\!+\!1\!+\!\log(2\pi)\right)\!+\!\log \left( \sum_{i=1}^M p_i^\star e^{-(\sqrt{x_i^\star}+\sqrt{\inr})^2} \right) 
\label{eq:indep}
\\
&+\int_{y\geq 0} K(x,y)\left(2\sqrt{y}(\sqrt{x_M^\star}+\sqrt{\inr})\right) \ dy, 
\label{eq:NO indep}
\end{align}
\label{eq:indep all together}
\end{subequations}
where the second term in~\eqref{eq:indep} is independent of $x$ 
and hence we only need to deal with~\eqref{eq:NO indep}.
The term in~\eqref{eq:NO indep} can be bounded as
\begin{align}
\mathbb{E}\left[\sqrt{\widetilde{Y}} \Big|\widetilde{X}=x\right]
&\leq \sqrt{\mathbb{E}\left[\widetilde{Y} \Big|\widetilde{X}=x\right]} \label{eq:Jensen}
=\sqrt{1+x+\inr},
\end{align}
where~\eqref{eq:Jensen} follows from Jensen's inequality 
and by~\eqref{eq:conditional output power}.
With the bound in~\eqref{eq:indep all together}
back into the KKT condition in~\eqref{eq:kkt} we get 
\begin{align}
- x+c\sqrt{x}+\kappa_1>-\lambda x +\kappa_2,
\label{eq:kkt at large x}
\end{align}
for some finite constants $c>0, \kappa_1, \kappa_2$ that are not functions of $x$. 
However, as $x\rightarrow \infty$, for $\lambda<1$ (by Proposition~\ref{prop:lambda}),
the right-hand-side of~\eqref{eq:kkt at large x} grows faster than the left-hand-side, which is impossible.
We reached a contradiction,  which implies that the optimal number of mass points can not be finite.
Thus, we ruled out case~\ref{finite}.

Having ruled out the possibility that $\Fopt$ has either infinitely many mass points in some bounded interval or is discrete with finitely many mass points, the only remaining option is that $\Fopt$ has countably infinitely many mass points, but only a finite number of masses in any bounded interval.
This concludes the proof. 
\end{IEEEproof}

 
 \section{Capacity at high INR}\label{sec: high inr}
 In this section, we prove that in the high INR regime, the communication system has only 1/2 the degrees of freedom compared to the interference-free system; which is a substantial improvement from the zero rate achieved when communication in presence of radar signal is prohibited. We also show that the Gaussian input is asymptotically optimal as $\inr \to \infty$.
 \begin{theorem}
 The capacity of channel~\eqref{eq:cap def main} as $\inr \to \infty$  is given by
 \[\lim_{\inr \to \infty}\CAP(\snr,\inr)= \frac{1}{2}\log(1+\snr).\]
 \end{theorem}
 \begin{IEEEproof}
 We show that in the high INR regime, the mutual information between the input and the output is upper bounded by $\frac{1}{2}\log(1+\snr)$ for {\it any} input distribution subject to an average~power constraint. We then show that the Gaussian input can asymptotically achieve this upper bound as $\inr \to \infty$. We write:
 \begin{subequations}
 \begin{align}
I(\bfX;\bfY)&=h(\bf{Y})-h(\bfW)\notag
\\
&=h(\widetilde{Y})-h(\widetilde{W})\label{eq:circ sym} \\
&\leq \int_{y\geq 0} f_{\widetilde{Y}}(y) \log \frac{1}{R(y)} \dy -h(\widetilde{W}),\label{eq:ub cross entropy}
 \end{align}
 \end{subequations}
 where~\eqref{eq:circ sym}  is because $\bfY$ and $\bfW$ are circularly symmetric, \eqref{eq:ub cross entropy} is due to non-negativity of relative entropy and where $R(y)$ is an auxiliary output density function which is absolutely continuous with respect to $f_{\widetilde{Y}}(y)$. Take
\begin{align}
R(y)=\frac{1}{\snr+1}e^{-\left(\frac{y+\inr}{\snr+1}\right)}I_0\left(2\frac{\sqrt{y\inr}}{\snr+1}\right),\label{eq:R(y)}
\end{align}
to be the auxiliary output distribution in~\eqref{eq:ub cross entropy}. The intuition behind this choice of $R(y)$ lies behind our conjecture that the Gaussian input is optimal for large INR and the fact that~\eqref{eq:R(y)} is the induced distribution on $\widetilde{Y}$ by a proper-complex Gaussian input.  Then by~\eqref{eq:ub cross entropy} we have
\begin{subequations}
\begin{align}
&\lim_{\inr\to \infty} I(\bfX;\bfY)
\leq\lim_{\inr\to \infty}  \int_{x\geq 0,y\geq 0}K(x,y) \log\left(\frac{(\snr+1)e^{\frac{y+\inr}{\snr+1}}}{I_0(2\frac{\sqrt{y\inr}}{\snr+1})} \right) \ dy \dF-\frac{1}{2}\log(4\pi e \inr) \label{eq:entropy high inr}
\\
&=\log(\snr+1) +\lim_{\inr\to \infty}\Big\{ \frac{\snr+2\inr+1}{\snr+1}-\int_{x\geq 0,y\geq 0}K(x,y)\log\left( \frac{e^{\frac{2\sqrt{y\inr}}{\snr+1}}}{\sqrt{4\pi\frac{\sqrt{y\inr}}{\snr+1}}}  \right)\ dy \dF \Big\}-\frac{1}{2}\log(4\pi e \inr) \label{eq:bessel substitution} 
\\
&=\frac{1}{2}\log(\snr+1) +\lim_{\inr\to \infty} \Big\{ \frac{\snr+2\inr+1}{\snr+1} -2\frac{\sqrt{\inr}}{\snr+1}\mathbb{E}[\sqrt{\widetilde{Y}}]+\frac{1}{4}\log(\inr)+\frac{1}{4}\mathbb{E}[\log(\widetilde{Y})]-\frac{1}{2}\log(e \inr)\Big\} \notag\\
&\leq \frac{1}{2}\log(\snr+1)+\frac{1}{2}+\lim_{\inr\to \infty}\Big\{ \frac{2\inr}{\snr+1}-\frac{2\sqrt{\inr}}{\snr+1}\left[\sqrt{\inr}+\frac{\snr+1}{4\sqrt{\inr}}+O(\frac{1}{\inr}) \right] \Big\}\label{eq:sqrt and log expectation}
\\
&=\frac{1}{2}\log(\snr+1)\label{eq: cap ub Gaus},
\end{align}
\end{subequations}
where~\eqref{eq:entropy high inr} is by calculating the entropy of a non-central Chi-square distribution with 2 degrees of freedom as the non-centrality parameter $\inr$ goes to infinity~\cite[eq. (9)]{lapidoth} and where~\eqref{eq:bessel substitution} and~\eqref{eq:sqrt and log expectation} are proved in Appendix~\ref{app:bessel substitution} and Appendix~\ref{app:sqrt and log expectation}, respectively.
Next, a Gaussian input can achieve the upper bound given in~\eqref{eq: cap ub Gaus}, as follows
\begin{subequations}
\begin{align}
    \lim_{\inr\to\infty}I(\bfX_G; \bfY) 
   &= \lim_{\inr\to\infty}h(\bfY) - h(\bfW) \label{eq:X_G}
\\&=\lim_{\inr\to\infty}
  \log(1+\snr)
 +h\left(\sqrt{\frac{\inr}{1+\snr}}e^{j\Theta_{\inr}}+{\bfZ}\right) 
 \!-h\left(\!\sqrt{\inr}e^{j\Theta_{\inr}}+{\bfZ}\!\right)\notag
\\&=\lim_{\inr\to\infty}
  \log(1+\snr)
 +\frac{1}{2}\log\left(1+\frac{\inr}{1+\snr}\right) 
 -\frac{1}{2}\log\left(1+\inr\right) \label{eq:entropy high inr2}
\\&=\frac{1}{2}\log(1+\snr)+
\lim_{\inr\to\infty}\frac{1}{2}\log\left(1+\frac{\snr}{1+\inr}\right) \notag
\\&=\frac{1}{2}\log(1+\snr),\notag
\end{align}
\end{subequations}
wher $\bfX_G$ is the proper-complex Gaussian input and where~\eqref{eq:entropy high inr2} is again by~\cite[eq. (9)]{lapidoth}.
 \end{IEEEproof}
 
\section{Numerical Evaluations}
\label{sec:num}
In this section, we numerically find a sub-optimal input for fixed $\snr =5$ and three different values of $\inr $ in the regime $0<\alpha :=\frac{\inr_\text{(dB)}}{\snr_\text{(dB)}} < 2$ and we compare the achieved rates with that of a proper-complex Gaussian input. We also evaluate different achievable rates in the regime $-1\leq \alpha\leq 2.5$ and compare them with the bound in Section~\ref{subsec:trivial bounds}. 

Numerically finding the optimal input for the channel considered in this paper is more challenging  compared to channels with finite dimensional capacity achieving inputs such as the ones considered in~\cite{smith} and \cite{bar-david}. In~\cite{smith}, for example, the optimization was initially performed for a very low SNR where an input with two mass points was proved to be optimal. As SNR increased, more mass points were added to the optimization problem in order to satisfy the KKT conditions and guarantee the optimality of the input. In the channel considered here however, a finite number of mass points is sub-optimal at {\it any} SNR. Hence, in the rest of this section, we find sub-optimal inputs with a finite number of mass points and solving the corresponding constraint optimization problem. We increase the number of mass points until the achieved rate remains unchanged after the $3$rd digit after the decimal point. Figure~\ref{fig:mass points snr} shows the location of the  mass points for each sub-optimal input as a function of SNR. We note that we do not claim the rates achieved with these inputs to be optimal, nor do we claim that these input distributions are capacity-optimal. It is however interesting to note that they can outperform Gaussian inputs. 

\begin{figure}[h]
\centering
\includegraphics[width=9cm]{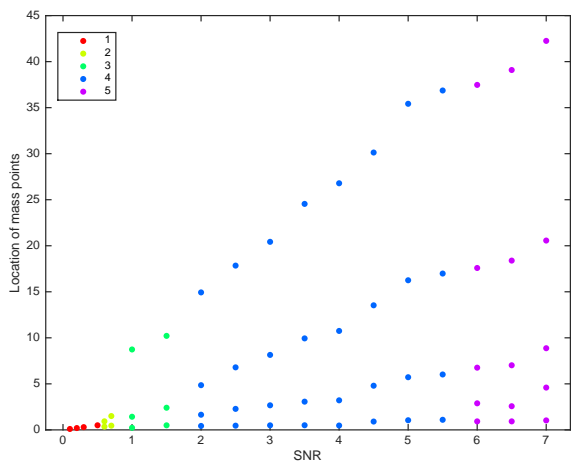}
\caption{Location of mass points for sub-optimal input as a function of SNR for fixed INR=$5$.}
\label{fig:mass points snr}
\end{figure}

We find the achievable inputs for fixed $\snr=5$ and three different values of $\inr=[3.6239, 9.5183,   25]$ which correspond to $\alpha=[0.8, 1.4, 2]$, by solving a finite dimensional constraint optimization problem. The achievable rates obtained by a Gaussian input and optimized finite dimensional inputs are given in Table~\ref{table: Gaussian vs opt input snr5}. As it can be seen, the optimized finite dimensional inputs  achieve marginally better rates than the Gaussian input.
\begin{table}[h]
\centering
       \caption{Achievable rates for Gaussian and optimized finite dimensional input, $\snr=5$.}
    \label{table: Gaussian vs opt input snr5}
    \begin{tabular}{|c|c|c|c|}
    \hline
    \diaghead{Optimized finite dimension}{Input}{INR}
     & $\snr^{0.8}=3.6239$ & $\snr^{1.4}=9.5183$ & $\snr^2=25$  \\ \hline
    Gaussian & 1.2905 & 1.1910 & 1.2470  \\ \hline
     Optimized finite dimension& 1.2927 & 1.1922 & 1.2480  \\ \hline
    \end{tabular}
 \caption{Achievable rates for Gaussian and optimized finite dimensional input, $\snr=10$.}
    \label{table: Gaussian vs opt input snr10}
 \begin{tabular}{|c|c|c|c|}
    \hline
    \diaghead{Optimized finite dimension}{Input}{INR}
     & $\snr^{0.8}=6.3096$ & $\snr^{1.4}=25.1189$ & $\snr^2=100$  \\ \hline
    Gaussian & 1.6986 & 1.6393 & 1.7100  \\ \hline
     Optimized finite dimension& 1.7108 & 1.6398 &  1.7102 \\ \hline
    \end{tabular}

\end{table}
 
\begin{figure}
\centering
\includegraphics[height=5cm]{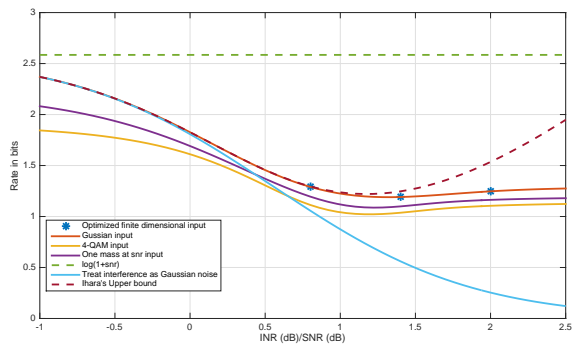}
\caption{\small Lower and upper bounds to the capacity vs $\inr$
for fixed $\snr = 5 = 6.9897 \text{dB}$.}
\label{fig:4-QAM-snr5dB}
\end{figure}

In Fig.\ref{fig:4-QAM-snr5dB} we plot achievable rates as function of $\inr$ for fixed $\snr=5$:
\begin{itemize}
\item (yellow solid line) 
an equally likely $4$-QAM constellation, 
\item (purple solid line)
a distribution with uniform phase and  only one mass point at $\sqrt{\snr}$ for the modulo, 
\item (orange solid line) 
a proper-complex Gaussian input, 
\item (blue solid line) 
treat the radar interference as Gaussian noise as given in~\eqref{eq:cap trivial lower}, and
\item (stared cyan points) optimized finite dimensional input.
\end{itemize}
We also show the outer bound in~\eqref{eq:cap trivial upper 2} (green dashed line)
and the one in~\eqref{eq:cap trivial upper 1} (red dashed line).

The Gaussian input performs very well for $\alpha :=\frac{\inr_\text{(dB)}}{\snr_\text{(dB)}} < 1$, where it closely follows the upper in~\eqref{eq:cap trivial upper 2}, in comparison to the discrete $4$-QAM input and a distribution with uniform phase and only one mass point at $\sqrt{\snr}$ for the modulo.  Although this behavior was expected for $\inr\ll 1$ (actually a Gaussian input is optimal for $\inr=0$), it is very pleasing to see that it actually performs very well for the whole regime $\inr\leq\snr$. 


We note that the equally likely $4$-QAM and the distribution with uniform phase and only one mass at $\sqrt{\snr}$ for the modulo are only a `constant gap' away from the the upper bound in~\eqref{eq:cap trivial upper 3} for the simulated $\snr=5$, which shows that capacity can be well approximated by inputs with a finite number of masses. The rate achieved by optimized finite dimensional input at $\inr=\snr^{0.8}, \snr^{1.4}$ and $\snr^2$ is only slightly higher than the rate achieved by a proper-complex Gaussian input.


\section{Conclusion}
\label{sec:concl}
In this paper we studied the structural properties of the optimal (communication) input of a new channel model which models the impact of a high power, short duty cycle, wideband, radar interference on a narrowband communication signal. In particular, we showed that the optimal input distribution has uniform phase independent of the modulo, which is discrete with countably infinite many mass points. 
We also argue that for large radar interference there is a loss of half the degrees of freedom compared to  the interference-free channel.

\section*{
Acknowledgment} The work of the authors was partially funded by NSF under award 1443967. The contents of this article are solely the responsibility of the authors and do not necessarily represent the official views of the NSF.

\section{Appendices}

\subsection{Derivation of the kernel $K(x,y)$ in~\eqref{eq:kernel def}}
\label{app:kernel derivation}

By~\eqref{eq:additive ch noise pdf} and
by passing to polar coordinates we have
\begin{subequations}
\begin{align}
K(x,y)
&:= f_{\widetilde{Y} \big| \widetilde{X}}(y|x)\notag
\\&= 
\int_{0}^{2\pi}d\phi  
\int_{0}^{2\pi}\frac{d\alpha}{2\pi}
\cdot
\frac{ e^{-|\sqrt{y}e^{j\phi}-\sqrt{x}e^{j\alpha}|^2-\inr} }{2\pi}
\notag \cdot
 I_0\left(2\sqrt{I}|\sqrt{y}e^{j\phi}-\sqrt{x}e^{j\alpha}|\right) \notag
\\&=
\int_{|\theta|\leq \pi}
\frac{e^{-(y+x+\inr-2\sqrt{yx} \cos (\theta))}}{2\pi}
 I_0\left(2\sqrt{\inr}\sqrt{y+x-2\sqrt{yx}\cos(\theta)} \right) d\theta \label{eq:kernel primary def}
\\&=
\int_{|\theta|\leq \pi}\frac{e^{-(y+x+\inr+2\sqrt{x\inr}\cos (\theta))}}{2\pi}
I_0\left(2\sqrt{y}\sqrt{x+\inr+2\sqrt{x\inr}\cos(\theta)}\right) d\theta,
\label{eq:kernel alternative def}
\end{align}
\end{subequations}
where \eqref{eq:kernel primary def} and~\eqref{eq:kernel alternative def} correspond to solving for the two integrals in different orders.
\subsection{The map $F_{\widetilde{X}} \to h(\widetilde{Y};F_{\widetilde{X}})$ is weak$^{\star}$ continuous}
\label{appendix:continuity proof}

To prove the weak$^{\star}$ continuity of the $h(\widetilde{Y};F_{\widetilde{X}})$ in~\eqref{eq:output entropy def}, we show that for any sequence of distribution functions $\{F_n\}_{n=1}^\infty \in \cal{F}$ if 
$F_{n}\overset{w^*}{\rightarrow}F_{\widetilde{X}}$
then
$h(\widetilde{Y};F_{n})\rightarrow h(\widetilde{Y};F_{\widetilde{X}})$.
We have
 \begin{subequations}
 \begin{align}
 \lim_{n\rightarrow \infty} h(\widetilde{Y};F_{n})
 &=\lim_{n\rightarrow \infty} \int_{y\geq 0} f_{\widetilde{Y}}(y;F_n)\log\frac{1}{f_{\widetilde{Y}}(y;F_n) }dy
 \notag
 \\
 &=\int_{y\geq 0} \lim_{n\rightarrow \infty} f_{\widetilde{Y}}(y;F_n)\log\frac{1}{f_{\widetilde{Y}}(y;F_n)} dy
 \label{eq:step a}\\
 &=h(\widetilde{Y};F_{\widetilde{X}}),
 \label{eq:step b}
 \end{align}
 \end{subequations}
where the exchange of limit and integral in~\eqref{eq:step a} is due to the Dominated Convergence Theorem~\cite{resnick2013probability}, and equality in~\eqref{eq:step b} is due to continuity of the map $F_{\widetilde{X}} \to f_{\widetilde{Y}}(y;F_{\widetilde{X}})\log f_{\widetilde{Y}}(y;F_{\widetilde{X}})$. 
This last assertion is true by noting that $x \to x\log x$ is a continuous function of $x\in\mathbb{R}_+$ and $f_{\widetilde{Y}}(y;F_{\widetilde{X}})$ in~\eqref{eq:output pdf def} is a continuous function of $F_{\widetilde{X}}$ since $K(x,y)$ in~\eqref{eq:kernel def} is a bounded continuous function of $x$ for all $y\in\mathbb{R}_+$.

Back to~\eqref{eq:step a}, to satisfy the necessary condition required in the  Dominated Convergence Theorem, we have to show that there exists an integrable function $g(y)$ such that  
 \begin{align}
 \mid f_{\widetilde{Y}}(y;F_n)\log f_{\widetilde{Y}}(y;F_n)\mid < g(y), \ \forall y\in \mathbb{R}_+.
 \label{eq:some integrable funct g(y)}
 \end{align}
We state the following Lemma which is a generalization of the one given in~\cite[Lemma A.2]{gursoytechnical}.
 \begin{lemma}
 \label{lemma}
  For any $\delta_1>0$ and $0<x\leq 1$
 \begin{align}
 0 \leq -x\log x \leq \frac{e^{-1}}{\delta_1}x^{1-\delta_1}.
 \label{eq:general xlogx bound}
 \end{align}
 \begin{proof}
 Fix a $\delta_1>0$; 
 the fuction $x \to -x^{\delta_1}\log x$ is concave in $0<x\leq 1$,
 and is maximized at $x=e^{-1/\delta_1}$. Hence
 $-x^{\delta_1}\log x\leq \frac{e^{-1}}{\delta_1}$ and~\eqref{eq:general xlogx bound} follows.
 \end{proof}
 \end{lemma}
According to Lemma~\ref{lemma} we can write
 \begin{align*}
 \mid f_{\widetilde{Y}}(y;F_n)\log f_{\widetilde{Y}}(y;F_n)\mid\leq \frac{e^{-1}}{\delta_1}f_{\widetilde{Y}}(y;F_n)^{1-\delta_1}.
 \end{align*}
We next need to find $\Phi(y) : f_{\widetilde{Y}}(y;F_n) \leq \Phi(y)$ which would then lead to 
\begin{align}
g(y)=\frac{e^{-1}}{\delta_1}\Phi(y)^{1-\delta_1},
\label{eq:integrable ub}
\end{align}
which is integrable for some $0<\delta_1$. 
Similarly to~\cite[eq. A9]{katz} we can show that for any $\delta_2>0$
\begin{align}
\Phi(y)=\left\{\begin{matrix}
1 &y\leq 16\inr \\ 
\frac{M}{y^{1.5-\delta_2} }& y>16\inr 
\end{matrix}\right.,
\label{eq:upper bound phi function}
\end{align}
is such a desirable upper bound 
for some $M<\infty$.
The proof is as follows.
For 
$y> 16\inr$ we  write
\begin{align}
f_{\widetilde{Y}}(y;F_{\widetilde{X}})&=\int_0^{(\sqrt{y}/4-\sqrt{\inr})^2} K(x,y)dF_{\widetilde{X}}(x)\label{eq:first range}
+\int_{(\sqrt{y}/4-\sqrt{\inr})^2}^\infty K(x,y) d F_{\widetilde{X}}(x).
\end{align}
The first term in~\eqref{eq:first range} can be upper bounded as
\begin{align}
  &\int_0^{(\sqrt{y}/4-\sqrt{\inr})^2} K(x,y)dF_{\widetilde{X}}(x) \notag
\\&
 \leq e^{-y} \int_0^{{(}\sqrt{y}/4-\sqrt{\inr}{)^2}}
\int_0^{2\pi} \frac{e^{-(x+\inr+2\sqrt{x\inr}\cos \theta)} }{2\pi} \cdot\notag
I_0\left(2\sqrt{y}(\sqrt{x}+\sqrt{\inr})\right) d\theta \ dF_{\widetilde{X}}(x)\notag
\\&\leq e^{-y} I_0\left(2\sqrt{y} \frac{\sqrt{y}}{4} \right)\notag
\int_0^{(\sqrt{y}/4-\sqrt{\inr})^2} \int_0^{2\pi}
\frac{e^{-(x+\inr+2\sqrt{x\inr}\cos \theta)} }{2\pi} d\theta \ dF_{\widetilde{X}}(x)\notag
\\&\leq e^{-y}I_0\left(y/2\right) \cdot 1 
   \leq e^{-y/2},
\label{eq:upper bound 1}
\end{align}
while the second term in~\eqref{eq:first range} can be upper bounded as
\begin{subequations}
\begin{align}
  \int_{(\sqrt{y}/4-\sqrt{\inr})^2}^\infty K(x,y)dF_{\widetilde{X}}(x)
  \notag
&\leq \mathbb{P}[\widetilde{X}>(\sqrt{y}/4-\sqrt{\inr})^2]
\notag
 \cdot
\frac{e^{-y}}{2\pi}\int_0^{2\pi} \sup_{x_\theta>0}\Big\{e^{-x_\theta}I_0\left(2\sqrt{y}\sqrt{x_\theta}\right)\Big\}d\theta
\notag
\\
&\leq \frac{e^{-y}}{2\pi}\int_0^{2\pi}  \ \snr \
\frac{\sup_{x_\theta>0}\Big\{e^{-x_\theta}I_0\left(2\sqrt{y}\sqrt{x_\theta}\right)\Big\}}{(\sqrt{y}/4-\sqrt{\inr})^2} d\theta
\label{eq:chebyshev}\\
&\leq {e^{-y}}  \frac{3}{2} 
\frac{e^{y}}{\sqrt{4\pi y}}\left[1+O(1/ y) \right]\frac{\snr}{(\sqrt{y}/4-\sqrt{\inr})^2},
\label{eq:katz inequality}
\end{align}
\end{subequations}
where $x_\theta:=x+\inr+2\sqrt{x\inr}\cos(\theta)$,
the inequality in~\eqref{eq:chebyshev} is from Markov's inequality,
and the one in~\eqref{eq:katz inequality} is by~\cite[eq.(E.6)]{Katz-thesis}. 
By~\eqref{eq:upper bound 1} and~\eqref{eq:katz inequality}, we have
\begin{align*}
f_{\widetilde{Y}}(y;{ F_n})\leq \frac{12\snr}{\sqrt{\pi}}\left[\frac{1}{y^{1.5}}+O(\frac{1}{y^{2.5}}) \right].
\end{align*}
{Hence, for any $0<\delta_2<1$ there exists some $M<\infty$ and $y^*_{\delta_2}$, such that 
\begin{align}
f_{\widetilde{Y}}(y;{ F_n})<\frac{M}{y^{1.5-\delta_2}},
\label{eq: M_delta}
\end{align}
for all $y\geq y^*_{\delta_2}$. We fix $\delta_2$ now. Due to continuity of the $f_{\widetilde{Y}}(y;{ F_n})$ for $y\in [16\inr, y^*_{\delta_2}]$, there exists an $M<\infty$ such that~\eqref{eq: M_delta} holds for all $y>16\inr$. The bound in~\eqref{eq: M_delta} together with the one in~\eqref{eq:output pdf bounds} gives 
\begin{align*}
f_{\widetilde{Y}}(y;{ F_n})\leq\Phi(y),
\end{align*}
for any $0<\delta_2<1$ and some $M<\infty$} and where $\Phi(y)$ was defined in~\eqref{eq:upper bound phi function}. Finally, one can find small enough $\delta_1$ and $\delta_2$ such that $g(y)$ given in~\eqref{eq:integrable ub}
 is integrable.

\subsection{The map $F_{\widetilde{X}} \to h(\widetilde{Y};F_{\widetilde{X}})$ is strictly concave }
\label{appendix:concavity proof}

The function $h(\widetilde{Y};F_{\widetilde{X}})$ in~\eqref{eq:output entropy def0} is concave in $f_{\widetilde{Y}}(y;F_{\widetilde{X}})$ in~\eqref{eq:output pdf def} (because $x \to -x\log(x)$ is).
Since $f_{\widetilde{Y}}(y;F_{\widetilde{X}})$ is an injective function of $F_{\widetilde{X}}$ (due to invertibility of the kernel as proved in Appendix~\ref{app:invertibility}), we conclude that $h(\widetilde{Y};F_{\widetilde{X}})$ is a strictly concave function of $F_{\widetilde{X}}$.

\subsection{The functional $ h(\widetilde{Y};F_{\widetilde{X}})-L(F_{\widetilde{X}})$, is weakly* differentiable at $\Fopt$}
\label{app:weak differentiability}

By using the definition of the functional derivative, we show that  
$h^{\prime}_{\Fopt}(\widetilde{Y};F_{\widetilde{X}})$ and 
$L^{\prime}_{\Fopt}(F_{\widetilde{X}})$ exist for all $F_{\widetilde{X}},\Fopt$ and hence $h(\widetilde{Y};F_{\widetilde{X}})-L(F_{\widetilde{X}})$ is weak$^{\star}$ differentiable. 

First, for $\theta\in[0,1],$ we define 
$F_\theta:=(1-\theta)\Fopt+\theta F_{\widetilde{X}}$ 
and then we find the weak$^{\star}$ derivative of 
$h(\widetilde{Y};F_{\widetilde{X}})$ at $\Fopt$ as follows
\begin{align}
  h^{\prime}_{\Fopt}(\widetilde{Y};F_{\widetilde{X}})
& =\lim_{\theta \to 0^+} \frac{1}{\theta}\left[
    h(\widetilde{Y};F_\theta)
   -h(\widetilde{Y};\Fopt)
  \right]
\notag
\\
  &=\lim_{\theta \to 0^+} \frac{1}{\theta}\int_{x\geq 0} \int_{y\geq 0} K(x,y) \log \frac{1}{f_{\widetilde{Y}}(y;F_\theta)} dy\ dF_\theta(x) 
\notag
\\
  &-\lim_{\theta \to 0^+} \frac{1}{\theta}\int_{x\geq 0} \int_{y\geq 0} K(x,y) \log \frac{1}{f_{\widetilde{Y}}(y;\Fopt)} dy\ d\Fopt(x)
\notag
\\
  &= \int_{x\geq 0} h(x;\Fopt) dF_{\widetilde{X}}(x)-h(\widetilde{Y};\Fopt)
\notag
\\
  &-\int_{y\geq 0} \lim_{\theta \to 0^+} \frac{1}{\theta}  f_{\widetilde{Y}}(y;F_\theta) \log \frac{f_{\widetilde{Y}}(y;F_\theta)}{f_{\widetilde{Y}}(y;\Fopt)}dy, 
  \label{eq:theta change}
  \end{align}
where the interchange of limit and integral in~\eqref{eq:theta change} is due to Dominated Convergence Theorem. By~\cite[Lemma 6]{guo-lemma}, we can write
  \begin{align}
  &\Big|\frac{f_{\widetilde{Y}}(y;F_\theta)}{\theta}\log\frac{f_{\widetilde{Y}}(y;F_\theta)}{f_{\widetilde{Y}}(y;\Fopt)}\Big|\leq 
 f_{\widetilde{Y}}(y;F_{\widetilde{X}})+ f_{\widetilde{Y}}(y;\Fopt)
 \notag \\
&- f_{\widetilde{Y}}(y;F_{\widetilde{X}})\log f_{\widetilde{Y}}(y;F_{\widetilde{X}})- f_{\widetilde{Y}}(y;F_{\widetilde{X}})\log  f_{\widetilde{Y}}(y;\Fopt)
\notag \\
&\leq f_{\widetilde{Y}}(y;F_{\widetilde{X}})+ f_{\widetilde{Y}}(y;\Fopt)+2 f_{\widetilde{Y}}(y;F_{\widetilde{X}})(y+\inr+\snr),
\label{eq:integrable lemma}
\end{align}
where the right hand side of~\eqref{eq:integrable lemma} is integrable.
In addition, the term given in~\eqref{eq:theta change} is vanishing by L'Hospital's Rule. Hence, the weak$^{\star}$ derivative is given by 
  \begin{align}
 h^{\prime}_{\Fopt}(\widetilde{Y};F_{\widetilde{X}})
 &=\int_{x\geq 0} h(x;\Fopt) dF_{\widetilde{X}}(x)-h(\widetilde{Y};\Fopt).\label{weak-derivative}
  \end{align}
It is also easy to show that
\begin{align}
  L^{\prime}_{\Fopt}(F_{\widetilde{X}}) = {L(F_{\tilde{X}})-L(\Fopt)},
&
\end{align}
{exists because of the linear{ity} of the power constraint.}
 %
  %
  %
\subsection{Equaivalence of KKT conditions in~\eqref{eq:kkt} to~\eqref{eq:equivalent optimization}}
\label{app:kkt-equivalence}

Let $\mathcal{E}^{\text{opt}}$ be the set of points of increase of the optimal input distribution $\Fopt$. Then
\begin{align}
\int_{x\geq 0} \left(h(x;\Fopt)-\lambda x\right)dF_{\widetilde{X}}(x)\leq h(\widetilde{Y};\Fopt)-\lambda \snr \label{only-if}
\end{align}
for all $F_{\widetilde{X}}\in \mathcal{F} $ if and only if
\begin{align}
h(x;\Fopt)&\leq h(\widetilde{Y};\Fopt)+\lambda(x-\snr), 
\ \forall x\in \mathbb{R}_+\label{if-1},
\\
h(x;\Fopt)&= h(\widetilde{Y};\Fopt)+\lambda(x-\snr), 
\ \forall x\in \mathcal{E}^{\text{opt}}.
\label{if-2}
\end{align}
The {\it if} direction is trivial since the derivative given in~\eqref{weak-derivative} has to be non-positive. 
To prove the {\it only if} direction, assume that~\eqref{if-1} is false. 
Then there exists an $\widetilde{x}$ such that 
\begin{align*}
h(\widetilde{x};\Fopt)>h(\widetilde{Y};\Fopt)+\lambda(\widetilde{x}-\snr).
\end{align*}
If $\Fopt$ is a unit step function at $\widetilde{x}$, then
\begin{align*}
\int_{x\geq 0}\left(h(x;\Fopt)-\lambda x \right)dF_{\widetilde{X}}(x)=h(\widetilde{x},\Fopt)-\lambda\widetilde{x}
>h(\widetilde{Y};\Fopt)-\lambda \snr,
\end{align*}
which contradicts~\eqref{only-if}.
Assume that~\eqref{if-1} holds but~\eqref{if-2} does not, i.e., 
there exists $\widetilde{x}\in \mathcal{E}^{\text{opt}}$: 
\begin{align}
h(\widetilde{x};\Fopt)<h(\widetilde{Y};\Fopt)+\lambda(\widetilde{x}-\snr).
\label{eq:kkt condition contradict}
\end{align}
Since all functions in~\eqref{eq:kkt condition contradict} are continuous in $x$, the inequality is satisfied strictly on a neighborhood of $\widetilde{x}$ indicated as $E_{\widetilde{x}}$.
Since $\widetilde{x}$ is a point of increase, the set $E_{\widetilde{x}}$ has nonzero measure, i.e., $\int_{E_{\widetilde{x}}} d\Fopt(x) = \delta > 0$; hence
\begin{align*}
h(\widetilde{Y};\Fopt)-\lambda \snr&=\int_{x\geq 0}\left( h(x;\Fopt)-\lambda x\right)\ d\Fopt(x)
\\
&=\int_{E_{\widetilde{x}}}\left( h(x;\Fopt)-\lambda x\right)\ d\Fopt(x)
\\
&+\int_{\mathcal{E}^{\text{opt}}\backslash E_{\widetilde{x}}}\left( h(x;\Fopt)-\lambda x\right)\ d\Fopt(x)\\
&{<} \delta(h(\widetilde{Y};\Fopt)-\lambda \snr)+(1-\delta)(h(\widetilde{Y};\Fopt)-\lambda \snr),
\end{align*}
which is a contradiction.

\subsection{Proof of Proposition~\ref{prop:lambda}}\label{app:proof prop lambda}

We prove that $\lambda$ can not be equal to $0$ or greater than or equal to $1$. By the Envelope Theorem~\cite{envelope-theorem} and the upper bound in~\eqref{eq:cap trivial upper 1}, having $\lambda\geq1$ is not possible.
The case $\lambda^\text{opt}(\snr)=0$ is unfeasible; 
if otherwise, the unique solution of~\eqref{eq:kkt} (uniqueness follows by invertibility of the integral transform in~\eqref{eq:output pdf def} as proven in Appendix~\ref{app:invertibility}) would induce the output pdf
\begin{align}
f_{\widetilde{Y}}(y;\Fopt)=\exp\{-h(\widetilde{Y};\Fopt)\}, \ \forall y\in\mathbb{R}_+,
\end{align}
which is not a valid pdf since it does not integrate to one.
Therefore we conclude that we must have $0<\lambda^\text{opt}(\snr)<1$.

\subsection{The function $z \to g(z,\lambda)$ is analytic}
\label{app:analyticity}
The analyticity of $g(z,\lambda), \ z\in\mathbb{C}_+$, follows from the analyticity of $h(z;F_{\widetilde{X}})$ on the same domain, where $h(x;F_{\widetilde{X}})$ was defined in~\eqref{eq:marginal output entropy def}.
In other words, we want to show that the function
\begin{align}
h(z; F_{\widetilde{X}}) = \int_{y\geq 0} K(z,y)\log\frac{1}{f_{\widetilde{Y}}(y;F_{\widetilde{X}})} \ dy, \ z\in\mathbb{C}_+,
\label{eq:above}
\end{align} 
is analytic. 
Note that the integrand in~\eqref{eq:above} is a continuous function on $\{z\in \mathbb{C}_+\} \times \{ y\in \mathbb{R}_+  \}$ and analytic for each $y$ so we use the Differentiation Lemma~\cite{Lang-complex} to prove the analyticity by proving that $h(x;F_{\widetilde{X}})$ is uniformly convergent for any rectangle 
$K:=\{z\in \mathbb{C}: 0\leq a\leq \Re(z)\leq b, -b\leq \Im(z)\leq b\}$ (since any compact set $K\in \mathbb{C}$ is closed and bounded in the complex plane). 
By~\eqref{eq:output pdf bounds} 
we have
$\left|\log  f_{\widetilde{Y}}(y;F_{\widetilde{X}})\right|
  \leq 
y+\inr+\beta_{F_{\widetilde{X}}},
$, and as a result we have
\begin{align}
 &|h(z;F_{\widetilde{X}})|
  \leq \int_{y\geq 0} |K(z,y)| \ |\log f_{\widetilde{Y}}(y;F_{\widetilde{X}})|\ dy
  \notag
  \\
 &\leq \int_{y\geq 0} \frac{1}{2\pi}\int_{|\theta|\leq \pi}
 \left|e^{-(z+y-2\sqrt{zy} \cos \theta +\inr)}\right|
  \cdot
    \left|I_0\left(2\sqrt{\inr(z+y-2\sqrt{zy}\cos \theta)}\right)\right| 
 \cdot \left|y+\inr+\beta_{F_{\widetilde{X}}}\right| \ d\theta \ dy
 \notag 
 \\
 &\leq \int_{y\geq 0} \frac{1}{2\pi}\int_{|\theta|\leq \pi}e^{-\Re(z+y-2\sqrt{zy} \cos \theta+\inr)} 
 \cdot 
   I_0\left(2\Re \Big\{\sqrt{\inr(z+y-2\sqrt{zy}\cos \theta)}\Big\}\right)\ \left(y+\inr+\beta_{F_{\widetilde{X}}}\right) \ d\theta\ dy\notag
 \\
 &\leq \int_{y\geq 0} \frac{1}{2\pi}\int_{|\theta|\leq \pi}e^{-\Re(y+z-2\sqrt{zy} \cos \theta +\inr)}\notag
 \cdot 
   e^{2\Re \{\sqrt{\inr(z+y-2\sqrt{zy}\cos \theta)}\}} \left(y+\inr+\beta_{F_{\widetilde{X}}}\right)\ d\theta\ dy\notag
 \\
 &= \int_{y\geq 0} \frac{1}{2\pi}\int_{|\theta|\leq \pi} e^{-(\sqrt{\Re(y+z-2\sqrt{zy} \cos \theta)}-\sqrt{\inr})^2} 
 \cdot 
\left(y+\inr+\beta_{F_{\widetilde{X}}}\right)\ d\theta\ dy.\label{eq:exp decreasing}
\end{align}
Since~\eqref{eq:exp decreasing} is exponentially decreasing in $y\in \mathbb{R}_+$, the integral is bounded,  concluding the proof.

\subsection{Invertibility of the integral transform in~\eqref{eq:output pdf def}}
\label{app:invertibility}

To prove the invertibility of the transform 
\begin{align}
\breve{g}(y)=\int_{x\geq 0}K(x,y)g(x)\ dx, \ y\in \mathbb{R}_+,
\label{eq:invertibility}
\end{align}
we will show that if
$\breve{g}(y) \equiv 0\text{ for all } y\in \mathbb{R}_+,$ then
$ g(x)\equiv 0 \text{ for all } x\in \mathbb{R}_+.$ 
From the invertibility of~\eqref{eq:invertibility}, also the integral transform $\int_{y\geq 0}K(x,y)g(y) dy$ is invertible due to the symmetry of the kernel $K(x,y)$ in $x$ and $y$.

We first define the following two integrals~\cite[eq(6.633) and eq(6.684)]{table_of_integrals}
\begin{align}
\int_0^\infty e^{-\alpha y}I_\nu(\beta\sqrt{y}) J_\nu(\gamma \sqrt{y})\ dy= \frac{1}{2\alpha} \exp\left(\frac{\beta^2-\gamma^2}{4\alpha} \right) J_0\left(\frac{\beta \gamma}{2\alpha}\right),
\notag\\
\Re \{\alpha\}>0, \Re \{\nu\}>-1,
\label{integral_y}\\
\int_0^\pi  (\sin \theta)^{2\nu}  \frac{J_\nu\left(\sqrt{\alpha^2+\beta^2-2\alpha \beta \cos \theta}\right)}{\left(\sqrt{\alpha^2+\beta^2-2\alpha \beta \cos \theta}\right)^\nu}\ d\theta
=2^\nu \sqrt{\pi}\Gamma\left(\nu+\frac{1}{2}\right) \frac{J_\nu(\alpha)}{\alpha^\nu}\frac{J_\nu(\beta)}{\beta^\nu},
\notag \\
 \Re\{\nu\}>-\frac{1}{2},
\label{integral_theta}
\end{align}
where $J_\nu(.)$ and $I_\nu(.)$ are the $\nu$-th order Bessel function of the first kind and $\nu$-th order modified Bessel function of the first kind, and where $\Gamma(.)$ is the Gamma function.

We next use~\eqref{integral_y} and~\eqref{integral_theta} as follows.
If $\breve{g}(y)= 0$ for all $y\geq 0$, then for all $\gamma\geq 0$ we have
\begin{align}
   & {\textcolor{white}\Longleftrightarrow} \int_0^\infty J_0(\gamma\sqrt{y}) \breve{g}(y)\  dy = 0 \notag
\\ &\Longleftrightarrow
 \int_0^\infty g(x) dx \int_0^\pi 
 J_0\left(\gamma\sqrt{x+\inr+2\sqrt{x\inr}\cos \theta}\right)\ d\theta=0 
\label{eq:use integral 1}
\\ &\Longleftrightarrow
\int_0^\infty g(x) J_0(\gamma \sqrt{x})J_0(\gamma \sqrt{\inr})\ dx=0
\label{eq:use integral 2}
\\&\Longleftrightarrow
\int_0^\infty g(z^2) J_0(\gamma z) z\ dz =0
\notag
\\&\Longleftrightarrow \mathcal{H}\{ g(z^2)\} =0, 
\label{eq:hankel}
\\&\Longleftrightarrow
 g(z^2) =0,\  \forall z\in \mathbb{R}_+,
 \notag
\\&\Longleftrightarrow 
g(x) =0,\  \forall x\in \mathbb{R}_+,
\notag
\end{align}
where~\eqref{eq:use integral 1} follows by~\eqref{integral_y}, \eqref{eq:use integral 2} by~\eqref{integral_theta},  
and where $\mathcal{H}\{g(z)\}$ in~\eqref{eq:hankel} denotes the Hankel transform~\cite{hankel} of the function $g(z)$.

\subsection{Justification of~\eqref{eq:bessel substitution}} \label{app:bessel substitution}
In order to show that
\begin{align}
 &\lim_{\inr \to \infty}\int_{x\geq 0, y\geq 0}
 K(x,y) \log I_0\left(2\frac{\sqrt{y\inr}}{\snr+1}\right)  \ dy\ dF_{X'}(x)\notag
 \\
& =\lim_{\inr \to \infty} \int_{x\geq 0,y\geq 0}  
 K(x,y) \log  \left(\frac{e^{2\frac{\sqrt{y\inr}}{\snr+1}} }{\sqrt{4\pi \frac{\sqrt{y\inr}}{\snr+1}}}\right)\ dy \dF,\notag
\end{align}
we make the variable change $y\inr=z$ and prove
\begin{align}
 &\lim_{\inr \to \infty}\int_{x\geq 0,z\geq 0}
 K(x,\frac{z}{\inr}) \log I_0\left(2\frac{\sqrt{z}}{\snr+1}\right)  \ \frac{dz}{\inr}\dF \notag \\
 &= \lim_{\inr \to \infty} \int_{x\geq 0,z\geq 0} 
 K(x,\frac{z}{\inr}) \log \left( \frac{e^{2\frac{\sqrt{z}}{\snr+1}} }{\sqrt{4\pi \frac{\sqrt{z}}{\snr+1}}}\right)\ \frac{dz}{\inr} \dF.\notag
\end{align}
For $B=\max\{1,\inr^{1/3}\}$, we can write
\begin{subequations}
\begin{align}
-&\lim_{\inr\to \infty}\frac{1}{\inr}\int_{x\geq 0}\int_{z\geq 0} K(x,\frac{z}{\inr}) \log \left( I_0\left(2\frac{\sqrt{z}}{\snr+1}\right)\right)  \ dz\ \dF\notag 
\\
&=-\lim_{\inr\to \infty}\frac{1}{\inr}\int_{x\geq 0}\int_{z=0}^B K(x,\frac{z}{\inr}) \log \left( I_0\left(2\frac{\sqrt{z}}{\snr+1}\right) \right) \ dz\ \dF\notag 
\\
&-\lim_{\inr\to \infty}\frac{1}{\inr}\int_{x\geq 0}\int_{z\geq B} K(x,\frac{z}{\inr}) \log \left(  I_0\left(2\frac{\sqrt{z}}{\snr+1}\right) \right)  \ dz \dF\notag
\\
&= -\lim_{\inr\to \infty}\frac{1}{\inr}\int_{x\geq 0}\int_{z=0}^B K(x,\frac{z}{\inr}) \log \left(  I_0\left(2\frac{\sqrt{z}}{\snr+1}\right) \right)  \ dz \dF\notag
\\
&-\lim_{\inr\to \infty}\frac{1}{\inr}\int_{x\geq 0}\int_{z\geq B} K(x,\frac{z}{\inr}) \log \left(  \frac{e^{2\frac{\sqrt{z}}{\snr+1}} }{\sqrt{4\pi \frac{\sqrt{z}}{\snr+1}}}\right)  \ dz \dF\label{eq:approximate bessel}
\\
&= \lim_{\inr\to \infty}\frac{1}{\inr}\int_{x\geq 0}\int_{z=0}^B K(x,\frac{z}{\inr}) \log\left( \frac{e^{2\frac{\sqrt{z}}{\snr+1}} }{\sqrt{4\pi \frac{\sqrt{z}}{\snr+1}}}.\frac{1}{I_0\left(2\frac{\sqrt{z}}{\snr+1}\right)}\right)  \ dz \dF \label{eq:sandwich theorem}
\\
&-\lim_{\inr\to \infty}\frac{1}{\inr}\int_{x\geq 0}\int_{z\geq 0} K(x,\frac{z}{\inr}) \log \frac{e^{2\frac{\sqrt{z}}{\snr+1}} }{\sqrt{4\pi \frac{\sqrt{z}}{\snr+1}}}  \ dz \dF,\notag
\end{align}
\end{subequations}
 where~\eqref{eq:approximate bessel} holds by approximation of Bessel function with large arguments.

To this end, if we prove that the limit in~\eqref{eq:sandwich theorem} is zero, then 
our proof is complete. 
In this regard, we find an upper and lower bound on~\eqref{eq:sandwich theorem} and show that they are both zero.  We can upper bound~\eqref{eq:sandwich theorem} as
\begin{subequations}
\begin{align}
& \lim_{\inr\to \infty}\frac{1}{\inr}\int_{x\geq 0}\int_{z=0}^B K(x,\frac{z}{\inr}) \log\left( \frac{e^{2\frac{\sqrt{z}}{\snr+1}} }{\sqrt{4\pi \frac{\sqrt{z}}{\snr+1}}}.\frac{1}{I_0\left(2\frac{\sqrt{z}}{\snr+1}\right)}\right)  \ dz \dF \notag
\\
&\leq \lim_{\inr\to \infty}\frac{1}{\inr}\int_{z=0}^B \log\left( \frac{e^{2\frac{\sqrt{z}}{\snr+1}} }{\sqrt{4\pi \frac{\sqrt{z}}{\snr+1}}}.\frac{1}{I_0\left(2\frac{\sqrt{z}}{\snr+1}\right)}\right)  \ dz\label{eq:kernel less 1}
\\
&\leq \lim_{\inr\to \infty}\frac{1}{\inr}\int_{z=0}^B \log\left( \frac{e^{2\frac{\sqrt{z}}{\snr+1}} }{\sqrt{4\pi \frac{\sqrt{z}}{\snr+1}}}\right)  \ dz\label{eq:bessel larger 1}
\\
&= \lim_{\inr\to \infty}\frac{1}{\inr}\int_{z=0}^B\left( 2\frac{\sqrt{z}}{\snr+1}-\log \sqrt{\frac{4\pi}{\snr+1}}-\frac{1}{4}\log(z)\right)  \ dz\notag
\\
&=\lim_{\inr\to \infty}\frac{1}{\inr} \times \left(\frac{4z^{3/2}}{3(\snr+1)}-z\log(\sqrt{\frac{4\pi}{\snr+1}})+z\log(z)-z \right)\big|_{z=0}^B=0,\notag
\end{align}
\end{subequations}
where~\eqref{eq:kernel less 1} is by~\eqref{eq:kernel bounds} and where~\eqref{eq:bessel larger 1} is due to the fact that zero order modified Bessel function is always larger than or equal to $1$.
In addition, we can lower bound~\eqref{eq:sandwich theorem} as
\begin{subequations}
\begin{align}
& \lim_{\inr\to \infty}\frac{1}{\inr}\int_{x\geq 0}\int_{z=0}^B K(x,\frac{z}{\inr}) \log\left( \frac{e^{2\frac{\sqrt{z}}{\snr+1}} }{\sqrt{4\pi \frac{\sqrt{z}}{\snr+1}}}.\frac{1}{I_0\left(2\frac{\sqrt{z}}{\snr+1}\right)}\right)  \ dz \dF \notag
\\
&\geq \lim_{\inr\to \infty}\frac{1}{\inr}\int_{x\geq 0}\int_{z=0}^BK(x,\frac{z}{\inr}) \log\left( \frac{1 }{\sqrt{4\pi \frac{\sqrt{z}}{\snr+1}}}\right)  \ dz \dF\label{eq:bessel inequality}
\\
&= \lim_{\inr\to \infty}\frac{1}{4\inr}\int_{x\geq 0} \int_{z=0}^B K(x,\frac{z}{\inr}) \log\left( \frac{1 }{z}\right)  \ dz \dF\notag\\
&= \lim_{\inr\to \infty}\frac{1}{4\inr}\int_{x\geq 0} \int_{z=0}^1 K(x,\frac{z}{\inr}) \log\left( \frac{1 }{z}\right)  \ dz \dF+\lim_{\inr\to \infty}\frac{1}{4\inr}\int_{x\geq 0} \int_{z=1}^B K(x,\frac{z}{\inr}) \log\left( \frac{1 }{z}\right)  \ dz \dF\notag
\\
&\geq \lim_{\inr\to \infty}\frac{1}{4\inr}\int_{x\geq 0} \int_{z=1}^B K(x,\frac{z}{\inr}) \log\left( \frac{1 }{z}\right)  \ dz \dF\label{eq:B larger 1}\\
&\geq \lim_{\inr\to \infty}\frac{1}{4\inr} \int_{z=1}^B  \log\left( \frac{1 }{z}\right)  \ dz \label{eq:log less 1}\\
&=\lim_{\inr\to \infty}\frac{1}{\inr} \times \left(-z\log(z)+z \right)\big|_{z=1}^B=0,\notag
\end{align}
\end{subequations}
where~\eqref{eq:bessel inequality} is by the inequality
$I_0(x)\leq e^{x},$ and where~\eqref{eq:B larger 1} and~\eqref{eq:log less 1} are true for the choice of $B$.
Since both lower and upper bounds on~\eqref{eq:sandwich theorem} are zero, our claim follows.
\subsection{Justification of~\eqref{eq:sqrt and log expectation}: Calculation of $\lim_{\inr \to \infty}\mathbb{E}[\sqrt{\widetilde{Y}}]$}   \label{app:sqrt and log expectation}
Here we calculate the expected value of the output modulo given by 
\begin{align}
\lim_{\inr\to \infty}\mathbb{E}[\sqrt{\widetilde{Y}}]=\lim_{\inr\to \infty}\int_{x\geq 0}\int_{y\geq 0}\sqrt{y}K(x,y)\dy \dF\notag
=C_1+C_2,
\end{align}
where 
\begin{align*}
C_1&=\lim_{\inr\to \infty}\int_{x= 0}^A\int_{y\geq 0}\sqrt{y}K(x,y)\dy \dF,\\
C_2&=\lim_{\inr\to \infty}\int_{x\geq A}\int_{y\geq 0}\sqrt{y}K(x,y)\dy \dF,
\end{align*}
and where we take $A=\inr^{(1-\delta)}$ for some $\delta>0$.
\begin{align}
C_1&=\lim_{\inr\to \infty}\int_{x= 0}^A \int_{0}^{2\pi} \int_{y\geq 0}\sqrt{y}\kernel \dy \ d\theta \dF \notag
\\
&=\lim_{\inr\to \infty} \int_{x= 0}^A \int_0^{2\pi}e^{-(\xb)}\Gamma(3/2)\frac{e^{\xb}\sqrt{\xb}}{\Gamma(3/2)}\notag 
\\
&\hspace{6cm}\cdot \left[ 1+\frac{(\frac{1}{2})(\frac{1}{2})}{\xb}+O\left(\frac{1}{\inr}\right)\right]\ d\theta \dF \label{eq:expectation sqrt chi-square}\\
&=\lim_{\inr\to \infty}\int_{x=0}^A \int_0^{2\pi}\left(\sqrt{\xb}+\frac{1}{4\sqrt{\xb}}+O\left(\frac{1}{\inr}\right)\right) \ d\theta \dF\notag
\\
&= \lim_{\inr\to \infty}\int_{x=0}^A \left(\sqrt{\inr}+\frac{x+1}{4\sqrt{\inr}}+O\left(\frac{1}{\inr}\right) \right)\dF \label{eq: exp cos elip}\\
&=\sqrt{\inr}+\frac{\snr+1}{4\sqrt{\inr}}+O\left(\frac{1}{\inr}\right)\label{eq:DCT 0-A},
\end{align}
where~\eqref{eq:DCT 0-A} is by the Dominated Convergence Theorem.
The equality in~\eqref{eq:expectation sqrt chi-square} is due to~\cite[eq(6.631)]{table_of_integrals}
\begin{align*}
\int_{x\geq 0} \sqrt{x} e^{-\alpha x}I_0(2\beta \sqrt{x})dx= \frac{\Gamma(\frac{3}{2})}{\alpha^{3/2}} {\,}_1F_1\left( 3/2,1,\frac{\beta^2}{\alpha}\right),
\end{align*}
where ${\,}_1F_1(a,b,x)$ is the confluent hypergeometric function~\cite[Chapter 13]{olver_handbook}. The series expansion of ${\,}_1F_1(a,c,x)$  for $x\to \infty$ is given by~\cite[Section 13.7]{olver_handbook}
\begin{align*}
_1F_1( a,c,x)&=\frac{\Gamma(c) e^x x^{a-c}}{\Gamma(a)}\sum_{n=0}^\infty \frac{(c-a)_n(1-a)_n}{n!}x^{-n}\\
&= \frac{\Gamma(c) e^x x^{a-c}}{\Gamma(a)}\left[ 1+\frac{(c-a)(1-a)}{x}+O\left(\frac{1}{x^2}\right)\right],
\end{align*}
where
$(a)_n:=a(a+1)\ldots(a+n-1)$.

In addition, the equality in~\eqref{eq: exp cos elip} is justified by~\cite{table_of_integrals}
\begin{align*}
\hspace{-.5cm}2\pi\mathbb{E}_\theta\left[\frac{2\pi}{4(\sqrt{\xb})}\right]&=\frac{1}{2(\sqrt{\inr}-\sqrt{x})}{\bf K}\left(-\frac{4\sqrt{x\inr}}{x+\inr-2\sqrt{x\inr}}\right)+\frac{1}{2(\sqrt{\inr}+\sqrt{x})}{\bf K}\left(\frac{4\sqrt{x\inr}}{x+\inr+2\sqrt{x\inr}}\right) ,
\\
\hspace{-.5cm}2\pi\mathbb{E}_\theta\left[\sqrt{\xb}\right]&=2(\sqrt{\inr}-\sqrt{x}){\bf E}\left(-\frac{4\sqrt{x\inr}}{x+\inr-2\sqrt{x\inr}}\right)+2(\sqrt{\inr}+\sqrt{x}){\bf E}\left(\frac{4\sqrt{x\inr}}{x+\inr+2\sqrt{x\inr}}\right) ,
\end{align*}
where 
\begin{align*}
{\bf K}(k^2)&=\frac{\pi}{2} \sum_{n=0}^\infty \left[\frac{(2n-1)!!}{(2n)!!}\right]^2k^{2n}=\frac{\pi}{2}\left[ 1+\frac{1}{4}k^2+\frac{9}{64}k^4+\frac{25}{256}k^6+\ldots\right],\\
{\bf E}(k^2)&=\frac{\pi}{2}\left[1-\sum_{n=1}^\infty \left[\frac{(2n-1)!!}{(2n)!!}\right]^2\frac{k^{2n}}{2n-1}\right]=\frac{\pi}{2}\left[1-\frac{1}{4}k^2-\frac{3}{64}k^4-\frac{5}{128}k^6-\ldots\right],
\end{align*}
are respectively the complete elliptic integral of the first and second kind. Hence
\begin{align}
\mathbb{E}_\theta\left[\sqrt{\xb}\right]
&=\frac{1}{2}(\sqrt{\inr}-\sqrt{x})\left[1+\frac{\sqrt{x\inr}}{(\sqrt{\inr}-\sqrt{x})^2}-\frac{3}{4}\frac{x\inr}{(\sqrt{\inr}-\sqrt{x})^4}+\ldots\right]\notag
\\
&+\frac{1}{2}(\sqrt{\inr}+\sqrt{x})\left[1-\frac{\sqrt{x\inr}}{(\sqrt{\inr}+\sqrt{x})^2}-\frac{3}{4}\frac{x\inr}{(\sqrt{\inr}+\sqrt{x})^4}-\ldots\right]\notag
\\ 
&=\frac{1}{2}\left[\sqrt{\inr} -\sqrt{x} +\frac{\sqrt{x\inr}}{\sqrt{\inr}-\sqrt{x}}-\frac{3}{4}\frac{x\inr}{(\sqrt{\inr}-\sqrt{x})^3}+\ldots \right]\notag
\\
&+\frac{1}{2}\left[\sqrt{\inr} +\sqrt{x} -\frac{\sqrt{x\inr}}{\sqrt{\inr}+\sqrt{x}}-\frac{3}{4}\frac{x\inr}{(\sqrt{\inr}+\sqrt{x})^3} +\ldots\right]\notag
\\
&=\frac{1}{2}\left[\sqrt{\inr}+\frac{x}{\sqrt{\inr}-\sqrt{x}}-\frac{3}{4}\frac{x\inr}{(\inr\sqrt{\inr}-3\inr\sqrt{x}+3\sqrt{\inr}x-x\sqrt{x})}+\ldots\right]\notag
\\
&+\frac{1}{2}\left[\sqrt{\inr}+\frac{x}{\sqrt{\inr}+\sqrt{x}}-\frac{3}{4}\frac{x\inr}{(\inr\sqrt{\inr}+3\inr\sqrt{x}+3\sqrt{\inr}x+x\sqrt{x})} +\ldots \right]\notag
\\
&=\sqrt{\inr}+\frac{x}{4\sqrt{\inr}}+O\left(\frac{1}{\inr}\right)\label{eq: estimate1},
\end{align}
and
\begin{align}
\mathbb{E}_\theta\left[\frac{1}{4(\sqrt{\xb})}\right]
&=\frac{1}{4}\frac{1}{(\sqrt{\inr}-\sqrt{x})}\left[1-\frac{\sqrt{x\inr}}{(\sqrt{\inr}-\sqrt{x})^2}+\ldots\right] \notag
\\
&+\frac{1}{4}\frac{1}{(\sqrt{\inr}+\sqrt{x})}\left[1+\frac{\sqrt{x\inr}}{(\sqrt{\inr}+\sqrt{x})^2}+\ldots\right]\notag
\\ 
&=\frac{1}{4}\left[\frac{1}{\sqrt{\inr} -\sqrt{x}} -\frac{\sqrt{x\inr}}{(\sqrt{\inr}-\sqrt{x})^3}+\ldots \right]\notag
\\
&+\frac{1}{4}\left[\frac{1}{\sqrt{\inr} +\sqrt{x}} +\frac{\sqrt{x\inr}}{(\sqrt{\inr}+\sqrt{x})^3}+\ldots \right]\notag
\\
&=\frac{1}{4\sqrt{\inr}}+O\left(\frac{1}{\inr}\right).\label{eq: estimate2}
\end{align}
Finally, given~\eqref{eq: estimate1} and~\eqref{eq: estimate2}, for $x\ll \inr$ (which holds in the region $0\leq x\leq A, A=\inr^{(1-\delta)}$) we have
\begin{align*}
&\mathbb{E}_\theta\left[\sqrt{\xb}+\frac{1}{4(\sqrt{\xb})}\right]=\sqrt{\inr}+\frac{x+1}{4\sqrt{\inr}}+O\left(\frac{1}{\inr}\right).
\end{align*}
As the result, and since $C_2\geq 0$
\begin{align*}
\mathbb{E}[\sqrt{\widetilde{Y}}]\geq C_1= \sqrt{\inr}+\frac{\snr+1}{4\sqrt{\inr}}+O\left(\frac{1}{\inr}\right).
\end{align*}

\subsection{Justification of~\eqref{eq:sqrt and log expectation}: Calculation of $\lim_{\inr \to \infty}\mathbb{E}_{\theta,\widetilde{X}}\left[\log(\widetilde{X}+\inr+2\sqrt{ \widetilde{X}\inr}\cos \theta)\right]$}
\begin{align*}
\lim_{\inr\to \infty}\mathbb{E}_{\theta,\widetilde{X}}\left[\log(\widetilde{X}+\inr+2\sqrt{ \widetilde{X}\inr}\cos \theta)\right]&=\lim_{\inr\to \infty}\int_{x\geq 0}\int_0^{2\pi}\log(\xb)d\theta \dF\\
&=C_1+C_2,
\end{align*}
where 
\begin{align*}
C_1&=\lim_{\inr\to \infty}\int_{x= 0}^\inr\int_0^{2\pi}\log(\xb)\ d\theta \dF,\\
C_2&=\lim_{\inr\to \infty}\int_{x\geq \inr}\int_0^{2\pi}\log(\xb)\ d\theta \dF.
\end{align*}
To calculate $C_1$, we state the following lemma.
\begin{lemma}\label{lemma:log int}
\begin{align}
\int_0^{2\pi} \log(1+2r\cos(x)+r^2)\ dx=0,\ 0\leq r\leq 1
\end{align}
\end{lemma}
\begin{proof}
Based on Cauchy's integral formula
\[f(a)=\frac{1}{2\pi j} \oint_\gamma \frac{f(z)}{z-a} dz,\]
for $0\leq r <1$, we can write
\begin{align*}
\int_0^{2\pi} \log(1+2r\cos(x)+r^2)\ dx&=\int_0^{2\pi} \log(1+re^{jx})\ dx+\int_0^{2\pi} \log(1+re^{-jx})\ dx\\
&=2\oint_\gamma \frac{\log(1+z)}{jz}\ dz=0.
\end{align*}
For $r=1$, we have
\begin{subequations}
\begin{align}
\int_0^{2\pi} \log(1+2r\cos(x)+r^2)\ dx&=4\pi\log(2)+4\int_0^\pi \log \left(\cos(\theta) \right)\ d\theta\label{eq:x}\\
&=4\pi \log(2) +4\int_0^{\frac{\pi}{2}}\log\left( \cos(\theta)\right)\ d\theta +4\int_0^{\frac{\pi}{2}}\log \left(\sin(\theta)\right)\ d\theta \notag\\
&=4\pi \log(2) +4\int_0^{\frac{\pi}{2}}\log \left(\frac{\sin(2\theta)}{2}\right)\ d\theta\notag\\
&=2\pi \log(2) +2\int_0^{\pi}\log \left(\sin(\theta)\right)\ d\theta\notag\\
&=2\pi \log(2) +2\int_0^{\pi}\log \left(\cos(\theta)\right)\ d\theta\label{eq:2x},
\end{align}
\end{subequations}
which according to~\eqref{eq:x} and~\eqref{eq:2x}, results in
\[4\pi\log(2)+4\int_0^\pi \log \left(\cos(\theta) \right)\ d\theta =2\pi \log(2) +2\int_0^{\pi}\log \left(\cos(\theta)\right)\ d\theta=0.\]
 %
\end{proof}
Based on lemma ~\ref{lemma:log int}, we see that
\begin{align*}
C_1&=\lim_{\inr\to \infty}\int_{x= 0}^\inr \int_0^{2\pi}\log(\xb)\ d\theta \dF,\\
&=\lim_{\inr\to \infty}\int_{x= 0}^\inr \log(\inr) \dF +\lim_{\inr\to \infty}\int_{x= 0}^\inr \int_0^{2\pi}\log(1+2\sqrt{\frac{x}{\inr}}\cos \theta+\frac{x}{\inr}) \ d\theta \dF\\
&=\log(\inr).
\end{align*}
In addition,
\begin{align*}
C_2&=\lim_{\inr\to \infty}\int_{x\geq \inr}\int_0^{2\pi}\log(\xb)\ d\theta \dF\\
&=\lim_{\inr\to \infty}\int_{x\geq \inr} \log(x) \dF \\
&\leq\lim_{\inr\to \infty}\int_{x\geq \inr} x \dF,
\end{align*}
which goes to zero as $\inr \to \infty$ by Dominated Convergence Theorem.
As the result
\[\lim_{\inr\to \infty}\mathbb{E}_{\theta,\widetilde{X}}\left[\log(\widetilde{X}+\inr+2\sqrt{ \widetilde{X}\inr}\cos \theta)\right] \to  \log\left(\inr\right).\]

\bibliography{refs}
\bibliographystyle{IEEEtran}

\end{document}